\documentclass[11pt]{article}
\usepackage[T1]{fontenc}
\usepackage{abraces}
\usepackage{graphicx}
\usepackage{color}
\definecolor{darkgreen}{rgb}{0,0.4,0}
\usepackage{authblk}
\usepackage[colorlinks,linkcolor=blue,citecolor=darkgreen]{hyperref}
\usepackage{amssymb}
\usepackage{amsthm}
\usepackage{graphicx}
\usepackage{float}
\usepackage{amsmath}
\usepackage{amsmath,bm,bbm}
\usepackage{color}
\usepackage{multirow}
\usepackage{mathrsfs}
\usepackage{lineno,hyperref}
\usepackage{lipsum}
\usepackage{stfloats}
\usepackage{titletoc}
\usepackage{appendix}
\usepackage{yhmath}
\usepackage{graphicx,amsmath}
\usepackage{tcolorbox}

\usepackage{color}
\definecolor{darkgreen}{rgb}{0,0.4,0}

\usepackage{enumerate}
\usepackage{graphicx}
\usepackage{amsmath,amssymb}
\usepackage{color}
\usepackage{lipsum}

\usepackage{geometry}

\newtheorem{remark}{\it Remark}

\newtheorem{definition}{Definition}

\newtheorem{assumption}{Assumption}
\newtheorem{corollary}{Corollary}
\newtheorem{proposition}{Proposition}
\newtheorem{example}{Example}

\def\begquo{\begin{quote}}
\def\endquo{\end{quote}}
\def\begequarr{\begin{eqnarray}}
\def\endequarr{\end{eqnarray}}
\def\begequarrs{\begin{eqnarray*}}
\def\endequarrs{\end{eqnarray*}}
\def\begarr{\begin{array}}
\def\endarr{\end{array}}
\def\begequ{\begin{equation}}
\def\endequ{\end{equation}}

\def\begdes{\begin{description}}
\def\enddes{\end{description}}
\def\begenu{\begin{enumerate}}
\def\begite{\begin{itemize}}
\def\endite{\end{itemize}}
\def\endenu{\end{enumerate}}

\def\lef[{\left[\begin{array}}
\def\rig]{\end{array}\right]}
\def\begcen{\begin{center}}
\def\endcen{\end{center}}
\def\begrem{\begin{remark}\rm}
\def\endrem{\end{remark}}
\def\begdef{\begin{definition}}
\def\enddef{\end{definition}}
\def\begpro{\begin{proposition}}
\def\endpro{\end{proposition}}
\def\begfac{\begin{fact}}
\def\endfac{\end{fact}}
\def\begass{\begin{assumption}}
\def\endass{\end{assumption}}

\def\begmat#1{\begin{bmatrix}#1\end{bmatrix}}

\def\cale{{\cal E}}
\def\cali{{\cal I}}

\def\calh{{\cal H}}
\def\calc{{\cal C}}

\def\calr{{\cal R}}

\def\cale{{\cal E}}

\def\call{{\cal L}}

\def\coloneqq{:=}

\def\bw{{\bar W}}

\def\L2e{{\cal L}_{2e}}

\def\rea{\mathbb{R}}

\def\diag{\mbox{diag}}

\def\col{\mbox{col}}
\def\hal{{1 \over 2}}

\def\diag{\mbox{diag}}

\def\min{{\mbox{min}}}

\usepackage[prependcaption,colorinlistoftodos]{todonotes}

\def\call{{\cal L}}
\def\calc{{\cal C}}

\def\cale{{\cal E}}
\def\calr{{\cal R}}

\def\hal{{1 \over 2}}

\def\col{\mbox{col}}

\def\L2{{\cal L}_2}
\def\L2e{{\cal L}_{2e}}

\def\rea{\mathbb{R}}

\def\diag{\mbox{diag}}

\def\begequarr{\begin{eqnarray}}
\def\endequarr{\end{eqnarray}}
\def\begequarrs{\begin{eqnarray*}}
\def\endequarrs{\end{eqnarray*}}
\def\begarr{\begin{array}}
\def\endarr{\end{array}}
\def\begequ{\begin{equation}}
\def\endequ{\end{equation}}

\def\begdes{\begin{description}}
\def\enddes{\end{description}}
\def\begenu{\begin{enumerate}}
\def\begite{\begin{itemize}}
\def\endite{\end{itemize}}
\def\endenu{\end{enumerate}}

\def\lef[{\left[\begin{array}}
\def\rig]{\end{array}\right]}
\def\begcen{\begin{center}}
\def\endcen{\end{center}}

\def\TAC{{\it IEEE Trans. on Automatic Control}}

\def\CDC{{\it IEEE Conf. on Decision and Control}}
\def\SCL{{\it Systems \& Control Letters}}
\def\AUT{{\it Automatica}}

\begin{document}

\title{\huge On necessary conditions of tracking control for nonlinear systems via contraction analysis}
\setcounter{footnote}{-1}

\author{Bowen Yi, Ruigang Wang and Ian R. Manchester\footnote{The authors are with Australian Centre for Field Robotics \& Sydney Institute for Robotics and Intelligent Systems, The University of Sydney, Sydney, NSW 2006, Australia (\texttt{bowen.yi@sydney.edu.au})} }

\newcommand{\Addresses}{{

}}

\providecommand{\keywords}[1]{{\textsc{Index terms---}} #1}
\date{\today \\ \vspace{0.2cm}
presented at the 59th IEEE Conference on Decision and Control, 2020}

\maketitle

\abstract{ In this paper we address the problem of tracking control of nonlinear systems via contraction analysis. The necessary conditions of the systems which can achieve universal asymptotic tracking are studied under several different cases. We show the links to the well developed control contraction metric, as well as its invariance under dynamic extension. In terms of these conditions, we identify a differentially detectable output, based on which a simple differential controller for trajectory tracking is designed via damping injection. As illustration we apply to electrostatic microactuators.}

\vspace{0.3cm}


%
\section{Introduction}
\label{sec1}
%
Trajectory tracking is one of the most important objectives in the area of motion control, particularly for autonomous systems, robotics and electromechanical systems, which is concerned with designing a feedback law to make the given system asymptotically follow a time parameterized path. The {\em de facto} standard technical route of tracking control is to translate it into a stabilization problem by defining an error dynamics, and then regulating to zero the induced nonlinear systems, with many nonlinear control techniques applicable. However, it brings the challenges to analyze nonlinear \emph{time-varying} systems.

An alternative route is to study control systems differentially along their solutions, which is widely known as contraction or incremental stability analysis \cite{ANDetal,FORSEP,LOHSLO,ANG}. The basic results on contraction analysis can be tracked back to the field of differential equations, see for example \cite{DEM,LEW}. It allows us to study the evolution of nearby trajectories to each other from an auxiliary linearized dynamics, the stability of which can be characterized by Finsler-Lyapunov functions with an elegant geometric interpretation \cite{FORSEP}. Nevertheless, most works on contraction theory are devoted to systems analysis, and in each case the corresponding constructive solutions are rarely discussed, with notable exception \cite{MANSLO,MANSLOcsl}. In \cite{MANSLO} control contraction metric (CCM) is introduced as a sufficient condition for exponential stabilizability of {\em all feasible trajectories} of given nonlinear systems, the notion of which resembles control Lyapunov functions (CLFs) for asymptotic regulation of nonlinear systems. The obtained controller based on CCM enjoys the benefit that the key constructive procedure can be formulated as off-line convex optimization. In \cite{CHAMAN}, the CCM method is extended to Finsler manifolds. On the application side the CCM technique has provided solutions to a wide variety of physical systems, see \cite{BAZSTE,SINetal} for applications to human manipulation and motion planning. 

In this paper, we present some further results on asymptotic tracking in the context of contraction analysis. The main contributions are twofold.

\begin{itemize}
\item[1)] Similarly to CLF which is necessary for asymptotic controllability, the CCM is also a necessary condition of universal asymptotic tracking for general nonlinear systems. We also consider the cases of robust tracking and with dynamic extension, and show that CCMs are invariant under dynamic extension.

\vspace{0.1cm}

\item[2)] Motivated by the proposed necessary conditions, we provide a simple differential controller design, {\em i.e.}, injecting damping along an elaborated differentially passive output. The design is smooth globally, unlike the one in \cite{MANSLO} excluding a zero-Lebesgue set.
\end{itemize}

The paper is organized as follows. In Section \ref{sec:2}, we give the problem formulation and some preliminaries on differential dynamics. Section \ref{sec:3} presents the main results of the paper on necessary conditions of universal asymptotic tracking from several perspectives. Based on them, we discuss the tracking controller design in Section \ref{sec:4}. Some examples are given in Section \ref{sec:5}, and then the paper is wrapped up by some concluding remarks.

\textbf{Notations.} All mappings are assumed smooth. For full-rank matrix $g\in \rea^{n\times m}$ ($m<n$), we denote the generalized inverse as $g^\dagger  = [g^\top g]^{-1} g^\top$ and $g_\bot$ a full-rank left-annihilator. Given a matrix $M(x)$, a function $V(x)$ and a vector field $f(x)$ with proper dimensions, we define the directional derivative as $\partial_f M(x)= \sum_{i}{\partial M(x) \over \partial x_i} f_i(x)$ and $L_f V$ as the Lie derivative of $V$. For a square matrix $A$, $\texttt{sym}\{A\}$ represents its symmetric part $(A+A^\top)/2$.

\section{Preliminary}
\label{sec:2}

Consider the nonlinear control system
\begin{equation}
\label{eq:NL}
\dot{x} = f(x) + B(x) u, 
\end{equation}
with states $x \in \rea^n$, input $u \in \rea^m$ and $x(0) =x_0$, where the input matrix $B(x) \in \rea^{n \times m}$ ($n>m$) is full rank. We denote its solution as $x(t) = X(t,x_0,u)$. The control target is to track a predefined trajectory $x_d(t)$ generated by 
\begin{equation}
\label{eq:target}
\dot{x}_d = f(x_d) + B(x_d) u_d(t),  \qquad x_d(0) = x_{d0}
\end{equation}
with input $u_d(t)$. Following standard practice in tracking control, we assume that the system \eqref{eq:target} is forward complete and define the feasible input set as
$
\cale_{x_{d0}} := \left\{
 u_d \in \call_\infty^m \cap \calc^1 | \exists
X(t,x_{d0},u_d),~ \forall t\ge 0
\right\}
$
for a given $x_{d0}$. To streamline the presentation, we recall some definitions first.

\begin{definition}\rm\cite{ANG}
Consider the system \eqref{eq:NL} under the control $u = \alpha(x,t)$, the solution of which is forward invariant in $\cale \subset \rea^n$. The closed-loop system in $\cale$ is

\noindent {(IAS)} incrementally asymptotically stable (or asymptotically contracting) if $\forall(x_1,x_2) \in \cale$,
$$
|X(t,x_1,\alpha(x_1,t)) - X(t,x_2,\alpha(x_2,t))| \le \kappa(|x_1-x_2|,t)
$$
holds for any $t\ge 0$ and some function $\kappa$ of class $\mathcal{K}\call$.

\noindent {(IES)} incrementally exponentially stable (or contracting) if the system is IAS with $\kappa(a,t) = k_1 e^{-k_2 t} a$ for some constants $k_1,k_2 >0$.
\end{definition}

\begin{definition}\rm
\label{def:lya}
For the system $\dot x = F(x,t)$, a function $V: \cale\times \cale \to \rea_+$ is called IAS (or IES) Lyapunov function if 
    \begin{equation}
    \label{cond:1}
        L_{F(x,t)} V(x,\xi) + L_{F(\xi,t)} V(x,\xi) <0
    \end{equation}
    [or $\le \lambda V(x,\xi)$], and for some $a_1,a_2>0$ satisfying
    \begin{equation}
    \label{cond:2}
    a_1 |x - \xi|^2 \le V(x,\xi) \le a_2 |x-\xi|^2.
    \end{equation}
\end{definition}
The IAS of the system $\dot{x}=F(x,t)$ is equivalent to the existence of an IAS Lyapunov function for set stability of $x=\xi$, by considering an auxiliary dynamics $\dot\xi =F(\xi,t)$ \cite{ANG}. We are interested in designing a feedback law such that
\begin{equation}\label{eq:convergence}
  \lim_{t\to\infty} |x(t) - x_d(t)| =0.
\end{equation}


\noindent {\bf Problem Formulation} For the systems \eqref{eq:NL} and \eqref{eq:target} with any $u_d(t) \in \cale_{x_{d0}}$ and $\forall x_{d0} \in \cale$, design a controller $u=\alpha(x,t)$ achieving {\bf i}) the IAS of the system \eqref{eq:NL} (or IES for exponential tracking) and {\bf ii}) invariance of $\{(x,x_d) \in \rea^{2n}| x=x_d\}$.

\begin{remark}\rm
The qualifier ``universal'' refers to target trajectories generated by arbitrary $u_d \in \cale_{x_{d0}}$ and $x_{d0}\in \rea^n$. Another well-studied formulation of trajectory tracking is to achieve \eqref{eq:convergence} for \emph{a class of} inputs $u_d$, which is expected to have weaker requirements on control systems. It, however, involves additional excitation assumptions on desired trajectories or equivalently on $u_d$ \cite{JIANIJ}. A similar issue appears in nonlinear observers, where the universal case is related to \emph{uniform observablilty} \cite{BES}. For weakly observable systems, persistent excitation of system trajectories is required to continue the observer design \cite{ORTetal}.
\end{remark}


%
%
For any $(x_{d0}, x_0) \in \rea^{2n}$, there exists a regular smooth curve $\bar{\gamma}:[0,1] \to \rea^n$ such that $\bar\gamma(0)= x_{d0}$, $\bar\gamma(1) = x_0$, and
\begin{equation}\label{eq:ineq1}
  \int_{0}^{1} \sqrt{ {\dot \gamma (s)}^\top M(\bar{\gamma}(s) ) \dot \gamma(s) }
  \le
  (1+\varepsilon) d(x_{d0}, x_0)
\end{equation}
for some $\varepsilon>0$. Considering the infinitesimal displacement $\delta x(t) = {\partial \over \partial s} X(t,\bar{\gamma}(s),u_s)|_{s=1}$ and $\delta u \coloneqq {\partial \over \partial s} u_s|_{s=1}$, the time derivative of which is given by
\begequ\label{diff_dyn}
\delta\dot{x} = A(x,u) \delta x + B(x) \delta u,
\endequ
with $A(x,u):= {\partial f (x)\over \partial x} + {\partial B(x) u \over \partial x}$. 

\section{Main Results of Necessary Conditions}
\label{sec:3}

In this section, we present the main results of the paper, that is, identifying necessary conditions of the systems which may achieve universal tracking, under several different assumptions. The links to CCMs will also be clarified. 
%
\subsection{Necessary Condition of Universal Tracking}
\label{sec:31}

Let us consider the basic case of universal tracking, in which we need the following. 

\begin{assumption}
\rm\label{ass:1}
Consider the system \eqref{eq:NL} and the target dynamics \eqref{eq:target} forward invariant in $\cale \subset \rea^n$ with any input $u_d$ and $x_{d0}$. There exists a feedback law $u = \alpha(x,t)$\footnote{The feedback $u=\alpha(x,t)$ may also depend on $x_d(t)$ and $u_d(t)$, and the ``time-varying'' form is adopted to show this point.} such that
\begin{itemize}
    \item[1)] The set $\{(x,x_d)\in \cale\times \cale|x = x_d\}$ is forward invariant.
    \item[2)] The system $\dot x = f(x) + B(x)\alpha(x,t) := F(x,t)$ is IAS (or IES) with the Lyapunov function in the sense of Definition \ref{def:lya}.
\end{itemize}
\end{assumption}

The above assumption characterizes the problem formulation of universal tracking in terms of incremental stability.

\begin{proposition}\rm
\label{prop:necessary}
If Assumption \ref{ass:1} holds, then there exists a symmetric matrix $2a_1 I_n \le M(x) \le 2a_2 I_n $ such that

\noindent{C1)} for any non-zero $v\in \rea^n$, we have
\begin{equation}
\label{eq:ccm}
\begin{aligned}
v^\top M(x)B(x)  = 0 
\quad \implies \quad
v^\top \Bigg[
\partial_f M(x)  
+ 2M(x) {\partial f(x)\over \partial x}
\Bigg] v   < 0
\end{aligned}
\end{equation}
[or $\le - \lambda v^\top M(x) v$ for IES] and the PDEs for $i=1,\ldots, m$
\begin{equation}
\label{pde:1}
  \partial_{B_i} M(x) + {\partial B_i(x) \over \partial x}^\top M(x) + M(x){\partial B_i(x) \over \partial x} =0.
\end{equation}
\noindent{C2)} The \emph{dual} differential system
\begin{equation}
\label{syst:dual}
\begin{aligned}
\dot{p}  = {\partial f(x) \over \partial x}^\top p, ~
y_p  =  [M(x)B(x)]^\top p
\end{aligned}
\end{equation}
is uniformly zero-state detectable (with exponential convergence speed for IES).
\end{proposition}
\begin{proof}
Considering the IAS case, we define ${\partial^2 V \over \partial \xi^2}(x,x) =M(x)$, which is motivated by \cite{SANPRA}. According to \eqref{cond:2}, it yields $V(x,x)=0$ and
$
2a_1 I_n \le M(x) \le 2a_2 I_n.
$
For any pair $(x,\xi) \in \cale \times \cale$, we parameterize $\xi$ as $\xi=x+ rv$ for any $v\in \rea^n$ with $|r|$ sufficiently small. We get
$$
{\partial V \over \partial x}(x, x+rv) F(x,t) + {\partial V \over \partial \xi}(x,x+rv) F(x+rv,t) <0.
$$

A \emph{necessary} condition to the above inequality is that the second-order terms in the Taylor expansion with respect to $r$ are negative, that is
\begin{equation}
\label{ineq:p1_1}
\begin{aligned}
{r^2 \over 2}
\Bigg[
{\partial \over \partial x}\Big( v^\top {\partial^2 V \over \partial \xi^2 } v \Big)\Big|_{(x,x)}
 +
{\partial \over \partial \xi}\Big( v^\top {\partial^2 V \over \partial \xi^2 } v \Big)\Big|_{(x,x)}\Bigg]F(x,t)
+ 2r^2 v^\top M(x) {\partial F(x,t) \over \partial x} v
 <0.
\end{aligned}
\end{equation}
According to the definition of $M(x)$, we have
$$
{\partial \over \partial x}(v^\top M v)
=
 {\partial \over \partial x}\Big( v^\top {\partial^2 V \over \partial \xi^2 } v \Big)\Big|_{(x,x)}
 +
{\partial \over \partial \xi}\Big( v^\top {\partial^2 V \over \partial \xi^2 } v \Big)\Big|_{(x,x)},
$$
thus the inequality \eqref{ineq:p1_1} becomes
\begin{equation}\label{ineq:p1_2}
  {\partial \over \partial x}(v^\top M(x) v) F + v^\top\Bigg[ {\partial F\over \partial x}^\top M(x) + M(x){\partial F\over \partial x} \Bigg] v <0,
\end{equation}
for any non-zero $v\in\rea^n$. This condition relies on the existence of a feedback $\alpha(x,t)$ satisfying the above inequality. 

Now we decompose the feedback $\alpha(x,t)$ into
\begin{equation}
\label{decomposition}
\alpha(x,t) = \alpha_0(x,t) + u_d.
\end{equation}
For any trajectory $x_d\in \cale$, we assume that $x=x_d$ is invariant in Assumption \ref{ass:1}. That is, for $x_0 = x_{d0}$, we have that 
$$
\dot{x} - \dot{x}_d   = f(x) + B(x) \alpha(x,t) - f(x_d) + B(x_d)u_d \quad
\implies \quad \alpha_0(x_d,t)=0,
$$
where we used the full rank of $B(x)$ in the last implication. Invoking the Lagrange reminder representation of the Taylor series expansion, we note that $\alpha_0(x,t)$ can be represented as $\alpha_0(x,t) = \alpha_1(x,t)(x-x_d)$ for some function $\alpha_1$. Substituting \eqref{decomposition} into the inequality \eqref{ineq:p1_2}, we have
$$
\begin{aligned}
& v^\top \Bigg[ \partial_{(f + B(\alpha_0 + u_d ))} M  + 2\Big({\partial f \over \partial x} +B { \partial \alpha \over \partial x} \Big)^\top M 
+ \sum_{i=1}^{n} \Big[{\partial B_i \over \partial x}^\top M + M {\partial B_i \over \partial x}
\Big](\alpha_0 + u_d)_i
\Bigg]v <0,
\end{aligned}
$$
which is satisfied uniformly for \emph{arbitrary} $u_d \in \cale_{x_{d0}}$ with $\forall x_{d0}\in\cale$ and $v\neq 0$, thus the PDEs \eqref{pde:1} hold. Then, we have
$$
v^\top \Bigg[\partial_f M(x) + 2\Big({\partial f(x) \over \partial x}^\top + {\partial \alpha(x,t)  \over \partial x}^\top B(x)^\top\Big) M(x)\Bigg]v <0
$$
for any non-zero $v\in\rea^n$, equivalently written as \eqref{eq:ccm}. 

Let us consider the necessary condition C2, in which we need to show that for the system \eqref{syst:dual}
\begin{equation}
\label{impl:detectable}
    y_p \equiv 0 \quad \implies \quad \lim_{t\to\infty} p(t) =0.
\end{equation}
Consider the Lyapunov function candidate $\mathcal{V}(x,p) = p^\top M(x) p$, the time derivative of which is
$$
\dot{\mathcal{V}} = p^\top \Bigg[\dot{M}(x) + {\partial f(x) \over \partial x}^\top M(x) + M(x){\partial f(x) \over \partial x}  \Bigg] p,
$$
where $x$ is generated by \eqref{eq:NL}. Consider the case $y_p \equiv 0$ and \eqref{eq:ccm}, we have $\dot{\mathcal{V}} < 0$ for any $p \neq0$, thus verifying  \eqref{impl:detectable}. The IES case can be proved \emph{mutatis mutandis}.
\end{proof}

\begin{remark}\rm
The condition C1 for universal asymptotic tracking resembles the ``stronge'' CCM proposed in \cite{MANSLO} but without a fixed contracting rate. We underscore that the PDE \eqref{pde:1} is also a necessity of \emph{differential passivity} \cite{VDS}. The condition C2 motivates us to construct tracking controllers with an observation that stabilizing the differential system can be translated into driving to zero the differential output.
\end{remark}
\begin{remark}\rm
As figured out in \cite{MANSLO}, the CCM resembles the CLF for asymptotic stabilization of nonlinear systems \cite{SON}. The existence of a CLF is, indeed, a \emph{necessary} condition of asymptotic controllability of nonlinear systems. Similarly, Proposition \ref{prop:necessary} verifies the CCM as necessity to achieve universal asymptotic tracking.
\end{remark}

\subsection{Dynamic Extension is Unnecessary}
The CCM was originally introduced for \emph{static} feedback control. On the other hand, dynamic feedback is a widely popular technique in feedback control for different purposes, {\em e.g.}, achieving relative degree, output feedback, performance enhancement and relaxing constraints. Particularly, it is widely recognized that dynamic extension may make a given nonlinear system achieve relative degree, then combining with feedback linearization we can design a dynamic controller to obtain an error system with linear time invariant dynamics, in order to be able to track any feasible trajectories \cite[Section 5.4]{ISI}. Therefore, a natural question relies on whether we can simply the necessary conditions by introducing dynamic extensions. Let us first consider the following example.

\begin{example}\rm
Consider the nonlinear system
\begin{equation}
\label{examp:dyn_ext}
\dot{x} = \begmat{-x_1 \\ x_4 \\ -x_3 + x_4 \\ 0}
+
\begmat{1 & 0 \\ x_3^2 +1 & 0 \\ 0 & 0 \\ 0 & 1}u,
\quad
y = \begmat{x_1 \\ x_2}
\end{equation}
with input $u \in \rea^2$. A simple solution to output tracking is via feedback linearization. Note that the system does not have relative degree with the given output mapping \cite[Section 5.4]{ISI}. However, we are able to achieve (vector) relative degree $[2,2]$ w.r.t. the new input $[v_1~~u_2]^\top$ by adding dynamic extension
$
u_1 = \xi , ~
\dot{\xi} = v_1,
$
and then use feedback linearization to solve the problem. It is easy to verify that the system enjoys a CCM by performing a change of input $u =\col((x_3^2+1)(-x_2 + v_1), v_2)$. It implies that a static feedback can achieve universal tracking for this example.
\end{example}

The above example shows that relative degrees are not fundamentally related to the universal stabilizability. We are now in position to show that dynamic extension is {\em unnecessary} to relax requirements in contraction analysis. Consider the objective that the system \eqref{eq:NL} asymptotically tracks the trajectory generated by the target system \eqref{eq:target} with an integral control\footnote{It can be extended to the more general cases, but we here adopt the basic case to streamline the underlying mechanism. See Remark \ref{rem:general_dyn_ext}.}, that is
\begin{equation}
\label{target_dyn_int}
\begin{aligned}
  u_d = \theta , \quad 
  \dot{\theta}  = f_c(x_d,\xi) + u_{\tt I}^d
\end{aligned}
\end{equation}
with the extended state $\theta\in \rea^{m}$ and $u_{\tt I}^d \in \rea^{m}$ involving in the integral action. We are interested in the necessary conditions of universal tracking $x_d$ generated by \eqref{target_dyn_int}.
\begin{proposition}\rm
Consider the system \eqref{eq:NL} in closed loop as
\begin{equation}
\label{dyn_closed}
 \begin{aligned}
    \dot{x}  = f(x) + B(x) u_{\tt K}, \quad
    \dot{x}_c  =  u_{\tt I},
 \end{aligned}
\end{equation}
with the dynamic feedback $(u_{\tt K}, u_{\tt I})$. Suppose that the system \eqref{dyn_closed} is IAS and forward complete with $x_d$ generated under \eqref{target_dyn_int} as a particular solution for any $u_{\tt I}^d$. Then there exists a metric $2a_1 I_n \le M(x) \le 2a_2 I_n$ satisfying the condition C1.
\end{proposition}

\begin{proof}
The condition C1 is equivalent to the existence of a dual metric $W(x) = M^{-1}(x)$ such that
\begin{equation}
\label{W}
\begin{aligned}
 & B_\bot^\top(x) \bigg( \partial_f W + {\partial f(x) \over \partial x} W + W {\partial f(x) \over \partial x}^\top \bigg) B_\bot(x) <0 \\
 & \partial_{B_i} W(x) - {\partial B_i(x) \over \partial x} W - W {\partial B_i(x) \over \partial x}^\top =0,
\end{aligned}
\end{equation}
for $i=1\ldots,n$, with $B_\bot(x)$ a full-rank left annihilator.

When we introduce the additional degree of freedom to design dynamic extensions, it is equivalent to verify the above condition for the extended dynamics 
$$
\dot{\chi} = \bar f(\chi) + \bar B(\chi)\begin{pmatrix}u_{\tt K} \\ u_{\tt I}
\end{pmatrix}	
$$ 
with $\chi:= \col(x,x_c)$, $\bar f(\chi) = \col(f(x, 0_{m \times 1})$ and $\bar B(\chi) =\diag(B(x) , I_{m})$. Since $x_c$ can be any feasible trajectories in the target system \eqref{target_dyn_int}, following the proof of Proposition \ref{prop:necessary} and using the dual property, the extended system should satisfy
\begin{equation}
\begin{aligned}
  & \bar B_\bot^\top(\chi) \bigg( \partial_{\bar f} \bar W(x,x_c) + {\partial \bar  f(\chi) \over \partial \chi} \bar W + \bar W {\partial \bar f(\chi) \over \partial \chi}^\top \bigg) \bar B_\bot(\chi) <0  \\
&  \partial_{\bar B_i} \bar W(x,x_c) - {\partial \bar B_i(\chi) \over \partial \chi} \bar W - \bar W {\partial \bar B_i(\chi) \over \partial \chi}^\top  =0,
\label{barW}
\end{aligned}
\end{equation}
for some $\bar{a}_1  I_{n + m} \le \bar W(x,x_c) \le \bar{a}_2 I_{n + m}$ with $\bar a_1, \bar a_2 >0$. We partition the matrix $\bar W(x,x_c)$ conformally as
$$
\bar{W}(x,x_c) = \begmat{\bw_x(x,x_c)& \bw_{xc}(x,x_c)\\\bw_{xc}(x,x_c) & \bw_c(x,x_c)},
$$
and note that $\bw_1(x,x_c)$ is also positive definite. Computing the $(1,1)$-block of the second equation in \eqref{barW} for $i=1\ldots,n$, we may get
$$
\partial_{B_i} \bw_x(x,x_c) - 2{\tt sym}
\left\{ {\partial B_i(x) \over \partial x} \bw_x \right\} = 0
$$
as a \emph{necessary} condition, where we used $\partial_{\bar B_i} \bw_x(x,x_c) = \partial_{ B_i} \bw_x(x,x_c)$ since the last $m$ elements in $\bar{B}_i(x)$ $(i=1,\ldots, n)$ are zeros.

It is clear that a feasible full-rank annihilator of $\bar B(\chi)$ is
$
\bar B_\bot(\chi) = \col(B_\bot(x), 0),
$
based on which we may get the $(1,1)$-block of the first inequality in \eqref{barW} as
$$
  B_\bot^\top(x) \bigg( \partial_f \bar W_x(x,x_c) + {\partial   f(x) \over \partial x} \bar W_x + \bar W_x {\partial  f(x) \over \partial x}^\top \bigg)  B_\bot(x) <0.
$$
Note that the above inequality holds for all $(x,x_c) \in \rea^n \times \rea^m$. Simply selecting
$$
M(x) = \bw^{-1}_x(x,0),
$$
and invoking the duality, we complete the proof.
\end{proof}

In the above analysis we show that we cannot weaken the necessary conditions via adding an integral action. 
\begin{remark}\rm
\label{rem:general_dyn_ext}
It is natural to consider the more general case of dynamic extension
$
 u  = \alpha(x,x_c) , ~
 \dot{x}_c  = \eta(x,x_c) + \gamma(x,x_c)v
$
with new input $v$ and $x_c \in \rea^{n_c}$. If we have a radically unbounded assumption on $\alpha(x,x_c)$ for fixed $x$, C1 is still a necessary condition. It shows the invariance of CCMs under dynamic extension. Note that the additional radical unboundedness assumption is used to force the PDEs \eqref{pde:1} to hold uniformly. 
\end{remark}

\begin{remark}\rm
Invoking the fact that every feedback linearizable system admits a CCM, we conclude that the system which can achieve relative degree via dynamic extension also has a CCM. Roughly speaking, if a nonlinear system can achieve universal asymptotic tracking with desired trajectories generated by a dynamic controller, then the system has a CCM. 
\end{remark}

\subsection{Necessary Condition for Robust Tracking}
\label{sec:robust}

Now we are carrying out the analysis for robust universal tracking control. Consider the closed loop
\begin{equation}
\label{syst_pert}
\dot{x} = f(x) + 
B(x)\alpha(x,t) + w(t)
\end{equation}
under the feedback $\alpha(x,t)$, in the presence of perturbation $w(t) \in \mathcal{L}_2^e$, which asymptotically {\em practically} tracks the trajectories of \eqref{eq:target} with any $u_d(t) \in \cale_{x_{d0}}$ and $x\in \rea^n$. With a slight abuse of notations, we denote the solution of \eqref{syst_pert} as $X_F(t,x_0,w(t))$ and $F(x,t):=f(x) +B(x)\alpha(x,t)$. To this end, we require that the target trajectory $x_d$ generated by \eqref{eq:target} is a particular solution of \eqref{syst_pert} in the absence of $w(t)$, and the closed loop is incrementally input-to-state stable (ISS), {\em i.e.},
$$
|X_F(t,\xi_1,w_1) - X_F(t,\xi_2,w_2)| \le \beta_1(|\xi_1 - \xi_2|,t) + \beta_2(|w_1 - w_2|_\infty),
$$
with $\beta_1 \in \mathcal{KL}$ and $\beta_2 \in \mathcal{K}_\infty$ for any pairs $(\xi_1,\xi_2) \in \rea^n \times \rea^n$. We need the following.

\begin{definition}
\label{def:iISS}\rm
A smooth function $V(x,\xi) :\rea^n \times \rea^n \to \rea_{\ge 0}$ is called an incremental ISS Lyapunov function if \eqref{cond:2} holds and there exists $\beta_3 \in \mathcal{K}_\infty$ such that $\forall w_1,w_2\in \mathcal{U} \subset \rea^n$ $\forall x_1,x_2\in\rea^n$, we have
$$
\beta_3(|x_1-x_2|) \ge |w_1 - w_2| \quad \implies \quad
 L_{F(x,t) + w_1} V(x,\xi) + L_{F(\xi,t) + w_2} V(x,\xi) < -\lambda V(x,\xi),
$$
with $\lambda >0$.
\end{definition}
\begin{proposition}
\label{prop:robust_necessary}\rm
If Assumption \ref{ass:1} holds but with an incremental ISS Lyapunov function in the sense of Definition \ref{def:iISS}, then there exists a metric $2a_1 I_n \le M(x) \le 2a_2 I_n $ such that 
\begin{equation}\label{lmi:iISS}
\begin{aligned}
\bar v^\top \col(M(x)B(x), 0_{n\times m}) & = 0 \quad \implies \quad
  \bar v ^\top \begmat{ \partial_fM + 2{\tt sym}\{M{\partial f\over \partial x}\}  + \lambda I_n &   M(x)  \\
   M(x) & -  \gamma_0 I_n  } \bar v & < 0
\end{aligned}
\end{equation}
for any non-zero $\bar v \in \rea^{2n}$ and some $\gamma_0>0$.
\end{proposition}
\begin{proof}
The proof is similar to the one of Proposition \ref{prop:necessary}, by selecting $M(x) = {\partial^2 V(x,x) \over \partial \xi^2}$. If $w_1 = w_2 =0$, this case recovers the results in Proposition \ref{prop:necessary}, and thus we have \eqref{pde:1}. The implication in Definition \ref{def:iISS} is equivalent to 
$$
{\partial V(x,\xi)\over \partial x} (F(x,t) + w_1) + {\partial  V(x,\xi) \over \partial \xi} (F(\xi,t) + w_2)
 < -\lambda V(x,\xi) + \beta_4(|w_1 - w_2|)
$$
with $\beta_4 \in \mathcal{K}_\infty$ \cite[Remark 2.4, pp. 353]{SONWAN}. If  $w_1=w_2= 0$, then the above inequality degenerates to the IES case studied in Proposition \ref{prop:necessary}, thus \eqref{eq:ccm} and \eqref{pde:1} also hold for this case. 

For any pairs $(x,\xi) \in \cale \times \cale$ and $(w_1,w_2) \in \mathcal{U} \times \mathcal{U}$, we parameterize 
$
\xi = x + rv,~w_1 = w_2 - 2re
$
with $|r|$ and $|e|$ sufficiently small. Focusing on the second-order term in the Taylor expansion with respect to $r$, we get the following necessary condition
$$
\begin{aligned}
{r^2 \over 2}
\Bigg[
{\partial \over \partial x}\Big( v^\top {\partial^2 V \over \partial \xi^2 } v \Big)\Big|_{(x,x)}
 +
{\partial \over \partial \xi}\Big( v^\top {\partial^2 V \over \partial \xi^2 } v \Big)\Big|_{(x,x)}\Bigg]F(x,t)
  + 2r^2 v^\top M(x) {\partial F(x,t) \over \partial x} v
 - 2r^2 v^\top {\partial^2 V \over \partial x\partial \xi}\Big|_{(x,x)}e \\
  < -\lambda v^\top {\partial^2 V\over \partial^2 x}\Big|_{(x,x)}v + \gamma_0 e^\top e,
\end{aligned}
$$
for some $\gamma_0 >0$. It can be written as
$$
\begmat{v \\ e}^\top
\begmat{\partial_f M + 2{\tt sym}\{M{\partial F\over\partial x}\} + \lambda M & M \\ M & - \gamma_0 I_n}
\begmat{v \\ e} <0
$$
for any $\col(v,e) \neq 0$, where we have used \eqref{cond:2} and ${\partial^2 V \over \partial \xi^2}|_{(x,x)} = - M(x)$. Cancelling the input $\alpha(x,t)$ from the above inequality, we may get the inequality \eqref{lmi:iISS}. 
\end{proof}

\begin{remark}\rm
In \cite[Theorem 2]{ANG}, it was shown that the above incremental ISS Lyapunov function is a sufficient and necessary condition to the incremental ISS property of \eqref{syst_pert}, assuming that $\mathcal{U}$ is compact and $\alpha(\cdot)$ is time invariant. It is interesting to observe that the condition \eqref{lmi:iISS} is nothing, but just the robust CCM proposed in \cite{MANSLOcsl} for nonlinear $\calh_\infty$ control, with the ``output'' $y=x$. The above analysis shows the necessary perspective of the robust CCM in \cite{MANSLOcsl}.
\end{remark}

\section{Further Results}
\label{sec:4}

%
\subsection{Stabilizing Differential System via Damping Injection}
\label{sec:41}

In this section, we discuss some further results of the presented necessary conditions, which are motivating to tracking controller design. In \cite{MANSLO}, a Sontag's type of differential feedback controller $\delta u$ is constructed in order to stabilize the infinitesimal displacement $\delta x$, thus achieving IES. However, the obtained differential controller cannot be guaranteed smooth at $\delta x= 0$, since the \emph{small control property} only guarantees continuity. Overcoming this drawback is one of the motivations.

\begin{assumption}
\label{ass:sufficient}\rm
Given the system \eqref{eq:NL} and a forward invariant target dynamics \eqref{eq:target} under $u_d$, there exists a metric $\underline p I \le M(x) \le \bar p I$ satisfying the condition C1 of the IES case.
\end{assumption}

Unlike the CLF, the CCM defined on Riemannian manifold enjoys a \emph{quadratic} form, making it possible to conduct a structural decomposition. The differential detectability condition C2 motivates us to design carefully an output injection, along which we can passivitify the differential system.

\begin{proposition}
\label{prop:passivity}\rm 
Consider the system \eqref{eq:NL} satisfying Assumption \ref{ass:sufficient}. Then, there exist globally defined smooth functions $\gamma:\rea^n \to \rea_{\ge0}$ such that the differential feedback controller
\begin{equation}
\label{diff_control}
    \delta u = -\gamma(x) \calr(x) \delta y + \delta v
\end{equation}
with the damping matrix $\calr(x)\coloneqq [(MB)^\top MB]^{-1} $ and the differential output
$
\delta y = [M(x)B(x)]^\top \delta x,
$
makes the system differentially passive with the input-output pair $(\delta v,\delta y)$. Furthermore, the damping injection $\delta v= - \gamma_0 \delta y$ with $\gamma_0>0$ makes the origin of the differential dynamics \eqref{diff_dyn} asymptotically stable, and
$
\delta v= - \gamma_0 \calr(x) \delta y
$
makes the origin exponentially stable.
\end{proposition}
\begin{proof}
For convenience, we denote $P(x) \coloneqq M(x)B(x)$, and decompose each infinitesimal displacement $\delta x$ into two parts, one of which is tangent to $P(x)$ denoted as $\delta x_p$, and the other $\delta x_v$ is orthogonal to $P(x)$, that is,
$
\delta x_p \coloneqq P(x)P^\dagger(x)\delta x,
$
and
$
\delta x_v \coloneqq \delta x - \delta x_p.
$
It is easy to verify
$
   \delta x_v^\top P(x)  = \delta x^\top [I - P(P^\top P)^{-1} P^\top] P =0.
$
We define a differential storage function as $V(x,\delta x) = {1\over 2}\delta x^\top M(x) \delta x$, the time derivative of which is
$$
\begin{aligned}
 \dot{\aoverbrace[L1R]{V(x,\delta x)}{} }
\le & \delta x^\top
\Bigg[
{1\over 2}\partial_f M(x) +  {\partial f(x) \over \partial x}^\top M(x)
\Bigg]\delta x
 +
\delta x^\top M(x) \delta u
\\
\le & ({1\over 4r} - {\lambda\over 2}) |\delta x_v|_{M(x)}^2 + \Big[ {r\over  \underline p} \Upsilon(x)^2 - \gamma(x)\Big]|\delta x_p|^2 +\delta v^\top \delta y
\end{aligned}
$$
where $u$ does not appear in the first inequality invoking \eqref{pde:1}, we have substituted $\delta x =\delta x_p + \delta x_v$ and used $\delta x_v^\top P=0$ in the second one with
$$
\Upsilon(x)\coloneqq
\left\|
\partial_f M(x) +2 {\tt sym}\Big\{{\partial f(x) \over \partial x}^\top M(x)\Big\}
\right\|,
$$
and in the last inequality we have used $\calr(x) P(x)^\top  = P^\dagger(x)$ and $\delta x_p^\top P(x) \delta v  = \delta v^\top \delta y$. For any $r>{1\over 2\lambda}$, by selecting smooth function
$$
    \gamma(x) \ge {r\over \underline p} \Upsilon(x)^2,~ \forall x\in \rea^n,
$$
we have
$$
{d\over dt}V(x,\delta x) \le \delta v^\top \delta y.
$$
It implies that the given system can be differentially passivitified via \eqref{diff_control}.

By adding a damping term $\delta v= -\gamma_0 \delta y$ with $\gamma_0 >0$, we have ${d\over dt}V(x,\delta) \le - |\delta y|^2$, thus
$$
\lim_{t\to\infty} \delta y(t) =0,
$$
in terms of Barbalat‘s lemma. In the proof of Proposition \ref{prop:necessary} we have shown that the condition C1 implies the zero-detectability of the differential system with the output $\delta y = P(x)^\top \delta x$. It implies that the origin of the differential system is exponentially stable.

For the case of $\delta v= - \gamma_0 \calr(x) \delta y$, we have
$$
\begin{aligned}
\dot{\aoverbrace[L1R]{V(x,\delta x)}{} } & \le
- \hal (\lambda - {1\over 2r})|\delta x_v|_{M(x)}^2
+
\delta x_v^\top P(x) \delta v \\
& \le - \lambda_0 V(x,\delta x)
\end{aligned}
$$
with $\lambda_0 = \min\{\lambda - {1\over 2r}, {2\over \underline p}\}$.
\end{proof}

The above analysis shows that the differential controller
\begin{equation}
\label{diff_control2}
 \delta u  = - [\gamma(x) + \gamma_0] \calr(x) \delta y \coloneqq K(x) \delta x
\end{equation}
can exponentially stabilize the differential dynamics, which is simply damping injection along the direction of the differentially zero-detectable output $\delta y$, identified in the condition C2. It guarantees ${d\over dt}V(x,\delta x) \le - \lambda_0 V(x,\delta x)$. 

\begin{remark}\rm 
The proposed differential controller enjoys global smoothness, which is simpler than the Sontag's type design in \cite{MANSLO}. The latter is not smooth in a zero Lebesgue measure set. In \cite{MANSLO} the well-known Finsler's Lemma is \emph{point-wisely} applied to calculate the differential controller in the form $\delta u = \rho(x) \delta y $ and the metric $M(x)$ simultaneously. Another difference between the proposed design and the one in \cite{MANSLO} relies on the involvement of a rotation matrix $\calr(x)$. 
\end{remark}

\subsection{Motivating Case and Path Integral}
\label{sec:42}

In this section, we study the construction of tracking controller $u=\alpha(x,t)$ complying with the differential controller $\delta u = K(x)\delta x$ proposed in Section \ref{sec:41}. In this subsection, we start from a motivating case, and invoke the well developed methods via path integral in \cite{MANSLO}. Our new analytical design will be introduced in Section \ref{sec:42}.

Let us come back the differential systems with the initial condition $x(0)= \bar\gamma(s)$, the corresponding differential controller at $X(t,\bar\gamma(s),u_s)$ is
\begin{equation}
\label{eq:diff_control_s}
  \delta u_s = K(X(t,\bar{\gamma}(s),u_s)) \delta x_s,
\end{equation}
The objective \eqref{eq:convergence} implies forward invariance, {\em i.e.}, $x(t) = x_d(t)$ for all $t \ge t_0$ if $x(t_0)= x_d(t_0)$. A necessary condition to it is the {\em boundary condition}
\begin{equation}\label{eq:boundary}
  u_s (t,\cdot)  \Big|_{s=0} = u_d(t).
\end{equation}
The differential feedback \eqref{eq:diff_control_s} may be rewritten as
\begin{equation}\label{eq:ode_s}
{\partial u_s(t, \cdot) \over \partial s} = K(X(t,\bar{\gamma}(s),u_s)) {\partial X (\cdot) \over \partial s}
\end{equation}
For a given moment $t>t_0$, the collection $\bar{\gamma}(\cdot)$ and a family of signals $u_s(\cdot) \in \call_\infty^m[0,t]$ for all $s \in [0,1]$, the solution $ X(t,\bar{\gamma}(\mu),u_s) =: \bar{\gamma}_x (t, \mu)$ defines a mapping
$$
\bar{\gamma}_x (t, \mu) : [0,\infty)\times [0,1] \to \cali_t,
$$
which is a smooth curve $\cali_t$ connecting $x(t)$ and $x_d(t)$ governed by \eqref{eq:NL}-\eqref{eq:target}. Along the curve $\cali_t$, considering the boundary condition \eqref{eq:boundary} and solving the ordinary differential equation (ODE) \eqref{eq:ode_s} at each moment $t\in [0,\infty)$\footnote{The differential equation \eqref{eq:ode_s} can be regarded as an ODE with respect to the variable $s$ with a given $t$.}, we get the desired control signal as
$$
\begin{aligned}
u(t,\cdot)  =u_s(t,\cdot)\Big|_{s=1} 
 =  u_d(t) + \int_{0}^{1}  K(\bar{\gamma}_x (t,\mu)) {\partial \bar{\gamma}_x(t,\mu) \over \partial \mu} d\mu.
\end{aligned}
$$

The implementation of the above controller design relies on calculating the mapping $\bar\gamma_x(t,\dot)$ numerically, which has a relatively heavy online computation burden. An alternative method is using the minimal geodesic $\gamma_{\tt m}(t,s)$ between $x(t)$ and $x_d(t)$ with the Riemannian metric $M(x)$. We have the following.

\begin{proposition}\rm
\label{prop:path_integral}
Consider the system \eqref{eq:NL} satisfying Assumption \ref{ass:sufficient}. For any $x_{d0}\in \rea^n$ and $u_d \in \cale$, the feedback controller
\begin{equation}
\label{control:path}
u = u_d(t) + \int_{0}^{1}  K(\bar{\gamma}_{\tt m} (x,x_d,\mu)) {\partial \bar{\gamma}_{\tt m}(x,x_d,\mu) \over \partial \mu} d\mu,
\end{equation}
exponentially achieve \eqref{eq:convergence}, where $K(\cdot)$ is defined in \eqref{diff_control2}, and the mapping $\gamma_{\tt m}$ is the minimal geodesic w.r.t. the metric $M(x)$ between $x(t)$ and $x_d(t)$, parameterized by $\mu \in [0,1]$.
\end{proposition}
\begin{proof}
The proof follows Proposition \ref{prop:passivity} and the proof of the third item in \cite[Theorem 1]{MANSLO}.
\end{proof}

See \cite{MANSLO,WANMAN} for more details about implementation, and \cite{LEUMAN} for the online computation of the minimal geodesic. This step is openly recongnized as the heaviest computational step of online realization.

If we make a change of variable to the ODE \eqref{eq:ode_s}, we may get the PDE
\begin{equation}\label{eq:pde_k1}
  {\partial \alpha(x,t) \over \partial x} = K(x),
\end{equation}
where we fix $u=\alpha(x,t)$. We denote $K_i(\cdot)$ as the $i$-column of $K(\cdot)$. The equation \eqref{eq:pde_k1} is only solvable if and only if
\begin{equation}
\label{pde:k}
{\partial K_i(x) \over \partial x} = \left[{\partial K_i(x) \over \partial x}\right]^\top.
\end{equation}

We have the following corollary, which is trivial to prove but motivating to our new development in the next subsection.

\begin{corollary}\rm
\label{cor:poincare}
Consider the system \eqref{eq:NL} satisfying Assumption \ref{ass:sufficient}. For any $x_{d0}\in \rea^n$ and $u_d \in \cale$, if we can find a smooth function $\gamma(x)$ guaranteeing \eqref{pde:k}, then the feedback controller
\begin{equation}
\label{control:poincare}
u = u_d(t) + \beta(x) - \beta(x_d)
\end{equation}
exponentially achieves \eqref{eq:convergence}, where $K(\cdot)$ is defined in \eqref{diff_control2}, and
$$
\alpha(x) = \int_{0}^x K(\chi)d\chi.
$$
\end{corollary}
\begin{proof}
The PDE \eqref{pde:k} guarantees the existence of $\beta(x)$. The infinitesimal displacement $\delta u_s$ at $X(t,\bar{\gamma}(s), u_s)$ is
$$
\delta u_s = K(X(t,\bar{\gamma}(s), u_s))\delta x_s.
$$
Selecting the differential Lyapunov function $\bar V(t,s) \coloneqq V(X(\cdot),\delta x_s) = \delta x_s^\top M(X(\cdot))\delta x_s$, the time derivative of which satisfies
$$
{d\over dt} \bar V(t,s) \le -\lambda_0 \bar V(t,s),
$$
since the mapping $K(\cdot)$ is constructed following Proposition \ref{prop:passivity}. Invoking the main results in \cite{FORSEP}, we conclude the incremental exponential stability of the closed-loop system under the controller \eqref{control:poincare}. Note the invariance of $x = x_d$, we achieve the universal tracking task \eqref{eq:convergence}.
\end{proof}

\vspace{0.1cm}

\begin{remark}\rm 
The selection of damping injection provides an additional degree of freedom to make the mapping $K$ satisfy \eqref{pde:k}. When the Riemannian metric $P(x)$ and the input matrix $g(x)$ are independent of state $x$, it is simple to complete controller design. However, it is a daunting task to guarantee \eqref{pde:k} in general cases.
\end{remark}

\subsection{The Controller Design}

The last step is the construction of tracking controller $u=\alpha(x,t)$ from the obtained differential feedback $\delta u$, which is solvable if and only if
$
\nabla K_i (x) = [\nabla K_i(x)]^\top
$
with $K_i(\cdot)$ as the $i$-column of $K(\cdot)$. Here, we propose an alternative method to (locally) realize the proposed differential controller. Indeed, the above-mentioned PDE is widely adopted in nonlinear observer design and adaptive control, see \cite{KARetal} for a recent review. We have the following. 
\begin{proposition}
\label{prop:dyn_impl}\rm
Consider the system \eqref{eq:NL} and the target system \eqref{eq:target} satisfying Assumption \ref{prop:passivity} with $x_d \in \call_\infty$. The system \eqref{eq:NL} is contracting under the dynamic feedback law
\begin{equation}
\label{eq:z}
   \dot{z} = f(x) + B(x) u - \ell (z -x), \quad u =   u_d + \beta(x,z) - \beta(x_d,z) 
\end{equation}
with $\ell>0$, $z(0),x_0 \in \mathcal{B}_\varepsilon(x_{d0})$ for $\varepsilon$ smaller than some $\varepsilon_\star >0$ and
$$
 \beta(x,z)  =  \int_{0}^{x_1} K_1(\mu, z_2, \ldots, z_n)d\mu + \ldots
+ \int_{0}^{x_n} K_m( z_1, \ldots, z_{n-1}, \mu)d\mu,
$$
thus achieving the task \eqref{eq:convergence}.
\end{proposition}
\begin{proof}
We give the sketch of proof. The dynamic extension \eqref{eq:z} is an IES system with $z=x$ a particular solution. If the system \eqref{eq:NL} is forward complete, then we have $\lim_{t\to\infty} |z(t) - x(t)| =0$. For convenience, we denote the feedback law in \eqref{eq:z} as 
$$
\alpha(x,t) \coloneqq u_d(t) + \beta(x,z(t)) - \beta(x_d(t),z(t)).
$$
 We then have
\begin{equation}\label{eq:partial}
\begin{aligned}
  {\partial \alpha \over \partial x}
  & =
  {\partial \beta(x,z) \over \partial x} 
   =
  \bigg[
  \begin{aligned}
  K_1(x_1, z_2, \ldots, z_n)~~ \Big|&~ \ldots&\Big| K_m(z_1,\ldots, z_{n-1}, x_n)
  \end{aligned}
  \bigg]
=:
  \hat{K}(x,z).
\end{aligned}
\end{equation}
Define
\begin{equation}
\label{eq:Delta}
\begin{aligned}
\Delta(x,z)  \coloneqq &\hat{K}(x,z) - K(x)
\end{aligned}
\end{equation}
satisfying
$
\Delta(x,x) =0.
$
If the closed-loop system \eqref{eq:NL} is forward complete, which can be shown for small $\varepsilon>0$, we have $\lim_{t\to\infty} \Delta(x(t),z(t)) = 0$ exponentially.

We now prove the contraction property of the closed-loop system by investigating its differential system along the solution $x =X(t,\bar\gamma(1),\alpha(x,t))$, which is
\begin{equation}
\label{diff_dyn_pert}
\delta \dot{x} = [A(x,u) + g(x)(K(x)+ \Delta(x,z))]\delta x.
\end{equation}
Invoking Proposition \ref{prop:passivity}, the above differential system can be regarded as an exponentially stable LTV system perturbed by a term $g(x)\Delta(x,z) \delta x$. Assuming that $x_0 \in B_\varepsilon(x_{d0})$ and $x_d(t) \in \call_\infty$ with small $\varepsilon>0$, we have $x(t) \in \call_\infty$. Therefore, the perturbation $g(x) \Delta(x,z) \delta x$ is an exponentially decaying term, and we conclude the exponential stability of \eqref{diff_dyn_pert} with some basic perturbation analysis \cite[Chapeter 9]{KHA}. Using the inverse Lyapunov theorem and \cite[Theorem 1]{FORSEP}, we are able to prove the (locally) IES of the control system \eqref{eq:NL} under the proposed feedback law.
\end{proof}
\section{Examples}
\label{sec:5}
%
\subsection{A Numerical Example}
In this subsection we consider a simple numerical example to verify the results in Section \ref{sec:43}, showing the relatively large domain of attraction. Consider the system
\begin{equation}
\label{eq:numerical}
\begin{aligned}
  \dot{x}_1 & = {1\over3}x_2^3 + x_2 \\
  \dot{x}_2 & = -x_2 + u.
\end{aligned}
\end{equation}
We may get the metric as $M = \begmat{3 & -1 \\ -1 & 2}$ with the differential controller $\delta u= K(x) \delta x$, where $K(x) = [-(x_2^2+1)~ ~-x_2^2]$. Constructing the dynamic extension and following the results in Subsection \ref{sec:42}, we may get the feedback law as
$$
u = - (z_2^2+1)(x_2) + (x_{d,2}^2 + 1)x_{d,2} - {1\over 3}(x_2^3 - x_{d,2}^3) + u_d(t).
$$
We compare it with the controller \eqref{control:path} by path integral via simulations. The initial conditions are $x_{d0}=[3~-1]^\top$, $x_0 = [-5 ~2]^\top$ and $z(0)=[0~0]^\top$, with $\ell =5$ and $u_d = \sin(t) - \cos(t)^2 x_{d,1}(t)$. We show the simulation result in Fig. \ref{fig:num}, where both the methods achieve IES. As expected, the proposed method has a larger overshoot at the beginning due to the dynamic extension, but reducing the online computation burden. We also test the controller with different initial conditions, illustrating that the domain of attraction is relatively large.

\begin{figure}[h]
    \centering
    \includegraphics[width=0.6\textwidth]{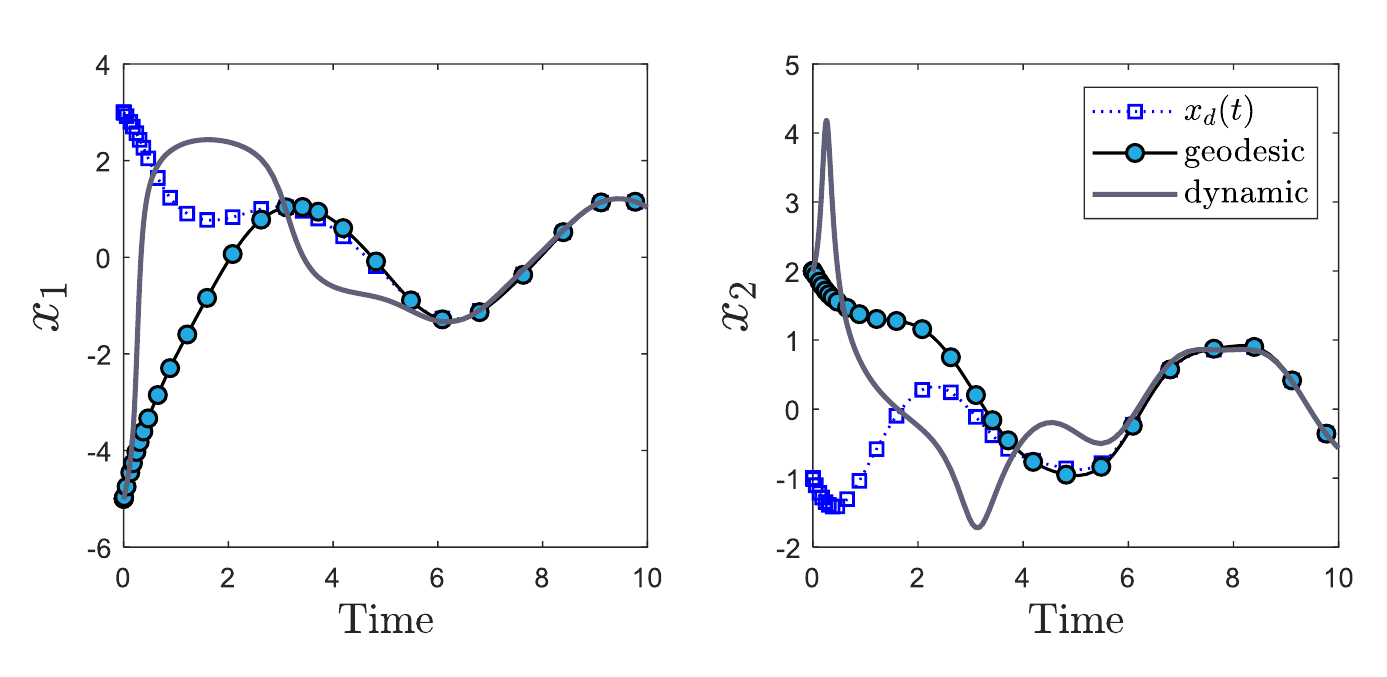}
    \includegraphics[width=0.3\textwidth]{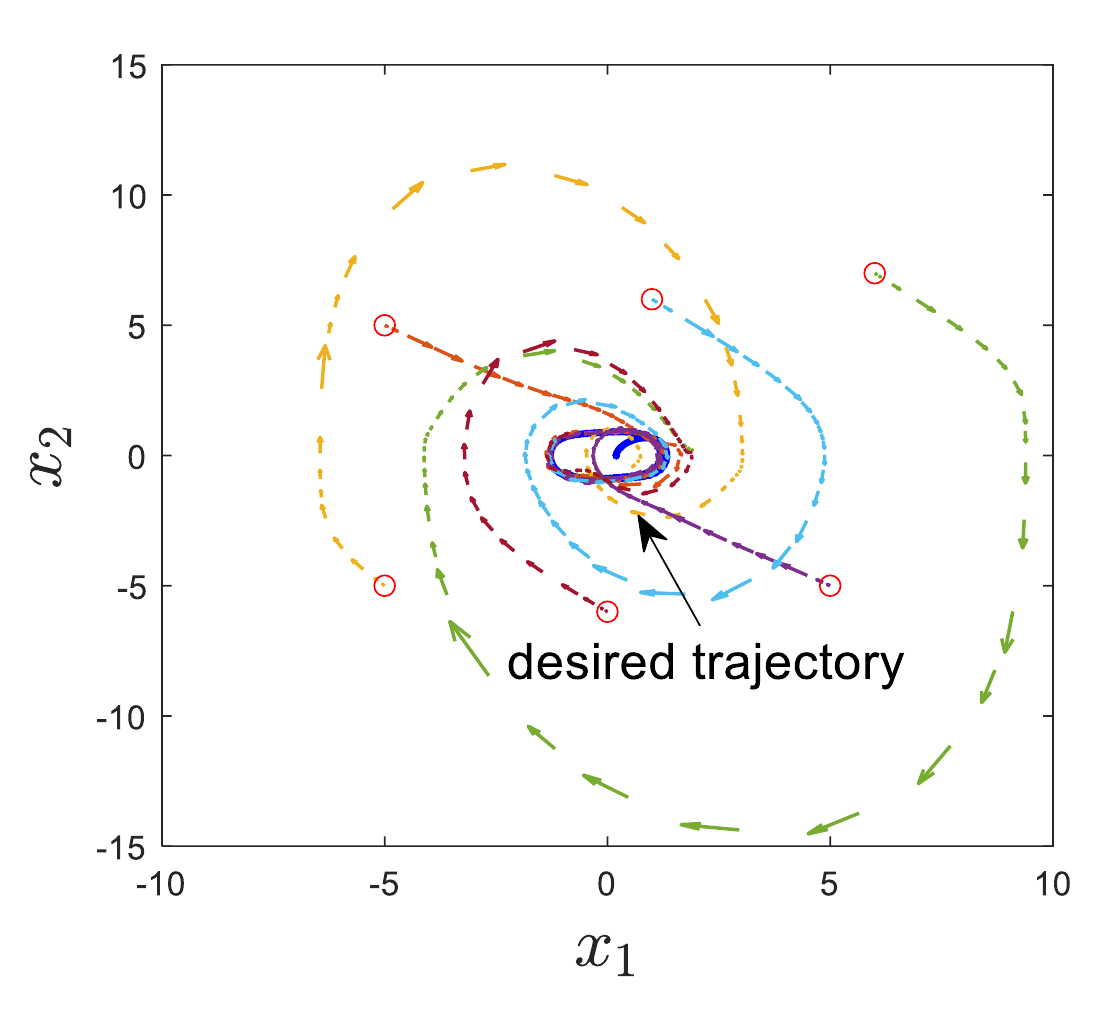}\caption{Simulation results of the numerical example}
    \label{fig:num}
\end{figure}

\subsection{Electrostatic Microactuator}
To illustrate the results, let us consider the problem of position tracking of the electrostatic microactuator, the model of which is given by \cite{MAIetal}
\begin{equation}
\label{EM_syst}
\begmat{\dot q\\ \dot p \\ \dot Q}
=
\begmat{{1 \over m} p \\ - k(q -1) - {1\over 2A\epsilon} Q^2 - {b\over m}p \\ -{1\over R A\epsilon}qQ+ {1\over R}u },
\end{equation}
and we denote $x:=\col(q,p,Q)$ representing the air gap, the momentum and the charge of the device. The systems state is defined on $\{(q,p,Q) \in \rea^3~ |~ 0\le q\le2, Q\ge 0\}$ due to physical constraints. Solving the inequality \eqref{eq:ccm}, we get a feasible solution 
$$
M^{-1} = \begmat{
{1\over 2bk} + {b^2 +km \over 2bk}& -{m\over 2}& 0\\
-{m \over 2} & {km^2 \over 2b} + {m \over 2b} & 0\\
0 & 0 & 1
},
$$
which is positive definite uniformly in the parameters $k>0, m>0$ and $b>0$. Noting that such example runs in a bounded state space, we can simply use a constant $\gamma>0$ for trajectory tracking. We give the simulation results in Fig. \ref{fig:EM} with normalized parameters $m=1,k=1,b=2,R=1$,$A = 3$ and $\epsilon = {1\over 2}$, and $x_{d0}=[0.2~0~0]^\top$ and $x_0=[1.5~1~2]^\top$. The control input of the target dynamics is selected as $u_d = \hal |\sin({1\over 5}t) + \cos(t)|$, and we fix $\gamma =2$. The simulation results validate the theoretical part.

\begin{figure}[h]
    \centering
    \includegraphics[width=0.8\textwidth]{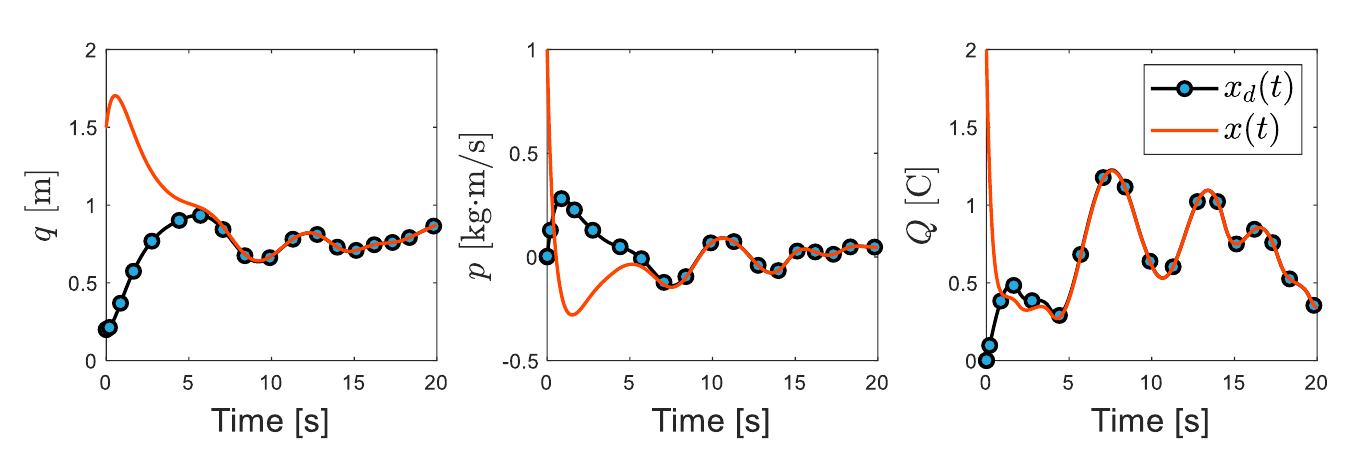}\caption{Simulation results for the model of electrostatic microactuator}
    \label{fig:EM}
\end{figure}

\section{Concluding Remarks}
\label{sec:6}
In this paper we have studied the necessary conditions of the systems which can achieve trajectory tracking with different cases, including universal asymptotic tracking, with dynamic extension and robust case. The invariance of CCMs under dynamic extension is clarified. We also show that the proposed differential detectability condition is intuitive for tracking controller design. The extensions in the following directions are of interests: 1) it is of practical interests to modify the results in Proposition \ref{prop:dyn_impl} in order to get a semi-global design; and 2) in this paper, we limit our attentions to the general nonlinear systems in the form \eqref{eq:NL}. For the systems with specific structures, it is promising to get more systematic constructive solutions.

%


\Addresses

\end{document}